\documentclass[a4paper,11pt]{article}
\pdfoutput=1 

\usepackage{jheppub} 
\usepackage{mathrsfs,graphicx,rotating,amsmath,amsfonts,mathtools,booktabs,amssymb,wasysym,caption}
\usepackage{amsthm}
\usepackage{hyperref}
\usepackage{slashed}
\usepackage{natbib}
\usepackage[table,xcdraw,dvipsnames]{xcolor}
\usepackage{graphicx}
\usepackage{bbold}
\usepackage[utf8x]{inputenc}
\usepackage[english]{babel}
\usepackage{multirow,multicol}
\usepackage{epstopdf}
\usepackage{bbm}
\usepackage{appendix}
\usepackage{libertine}
\usepackage{braket}
\usepackage{wasysym}
\usepackage{empheq}
\usepackage{cancel}
\usepackage{enumitem}
\usepackage{mathrsfs}
\usepackage{lmodern}
\usepackage{tabularx}
\usepackage{multicol}
\usepackage{color}
\usepackage{mathtools}

\newtheorem*{theorem}{Theorem}

\newcommand{\D}{{\rm d}}

\newcommand{\I}{{\rm i}}

\title{\boldmath An analytic approach to quasinormal modes for coupled linear systems}

\author[a]{Lam Hui,}
\author[a]{Alessandro Podo,}
\author[b]{Luca Santoni,}
\author[c,d]{Enrico Trincherini}

\affiliation[a]{Center for Theoretical Physics and Department of Physics,
Columbia University, New York, \\ 538 West 120th Street, NY 10027, USA}
\affiliation[b]{ICTP, International Centre for Theoretical Physics,
 Strada Costiera 11, 34151, Trieste, Italy}
\affiliation[c]{Scuola Normale Superiore, Piazza dei Cavalieri 7, 56126, Pisa, Italy}
\affiliation[d]{Istituto Nazionale di Fisica Nucleare (INFN) - Sezione di Pisa \\ Polo Fibonacci Largo B. Pontecorvo, 3, I-56127 Pisa, Italy}

\emailAdd{lh399@columbia.edu}
\emailAdd{ap3964@columbia.edu}
\emailAdd{lsantoni@ictp.it}
\emailAdd{enrico.trincherini@sns.it}

\abstract{Quasinormal modes describe the ringdown of compact objects deformed by small perturbations. In generic theories of gravity that extend General Relativity, the linearized dynamics of these perturbations is described by a system of coupled linear differential equations of second order. We first show, under general assumptions, that such a system can be brought to a Schr\"odinger-like form. We then devise an analytic approximation scheme to compute the spectrum of quasinormal modes. We validate our approach using a toy model with a controllable mixing parameter $\varepsilon$ and showing that the analytic approximation for the fundamental mode agrees with the numerical computation when the approximation is justified. The accuracy of the analytic approximation is at the (sub-) percent level for the real part and at the level of a few percent for the imaginary part, even when $\varepsilon$ is of order one. Our approximation scheme can be seen as an extension of the approach of Schutz and Will~\cite{Schutz:1985km} to the case of coupled systems of equations, although our approach is not phrased in terms of a WKB analysis, and offers a new viewpoint even in the case of a single equation.}

\begin{document}
\maketitle
\flushbottom

\section{Introduction}

Quasinormal modes (QNMs) describe damped vibrational modes of dissipative systems~\cite{Ching:1998mxl}, formally associated with a non-hermitian boundary value problem. They are essential tools in the study of the dynamics of astrophysical compact objects~\cite{Kokkotas:1999bd} and they encode hydrodynamical properties of strongly coupled systems through the gauge/gravity duality~\cite{Berti:2009kk}.
The study of quasinormal modes of black holes was initiated long ago in the early 1970s with the works of Vishveshwara~\cite{Vishveshwara:1970zz}, Press~\cite{Press:1971wr}, Teukolsky~\cite{Teukolsky:1974yv}, Chandrasekhar and Detweiler~\cite{Chandrasekhar:1975zza} (see for instance~\cite{Chandrasekhar:1985kt} for an introduction), and efficient computational methods are available for the well-known black hole solutions of pure General Relativity (GR)~\cite{Leaver:1985ax,Konoplya:2019hlu,Hatsuda:2019eoj}. 
Renewed interest has been spawned by experimental detection ~\cite{LIGOScientific:2020tif} which allows the observational test of theories that extend or modify GR. 
This includes the study of QNMs in selected theories beyond GR~\cite{Blazquez-Salcedo:2016enn,Blazquez-Salcedo:2017txk,Pierini:2021jxd,Wagle:2021tam,Srivastava:2021imr,Langlois:2022eta}, the development of effective field theories to parametrize and test extensions of GR~\cite{Endlich:2017tqa,Tattersall:2017erk,Franciolini:2018uyq,Hui:2021cpm,Cano:2020cao,Cano:2021myl,Mukohyama:2022enj,Mukohyama:2022skk,Khoury:2022zor} and
the proposal of new non-perturbative methods that relate the perturbation theory of explicit black hole solutions to Seiberg--Witten gauge theories~\cite{Aminov:2020yma,Bianchi:2021xpr,Bonelli:2021uvf,Bianchi:2021mft}. Recent developments in the understanding of the validity of perturbation theory further raise the hope of precision measurements in the future 
\cite{Giesler:2019uxc,Isi:2019aib,Mitman:2022qdl,Lagos:2022otp,Cheung:2022rbm}.

In this work we present a new analytic approximation scheme for the computation of the spectrum of QNMs of a system of (in general $N$) coupled linear differential equations in Schr\"odinger-like form. Our approach, based on a generalized adiabatic theorem, will also offer a new viewpoint on the regime of validity of the approximation of Schutz and Will~\cite{Schutz:1985km}.
Moreover we shall prove that, when an action formulation is available, it is in general possible to put the system of linear equations for the perturbations of a compact object in Schr\"odinger-like form and we shall describe an algorithm to do so, starting from the quadratic action in frequency space, performing an angular decomposition and some field redefinitions. 
The problem of finding the QNMs for a system of coupled linear differential equations has received significant attention in the recent literature. Alternative approaches to their study have been devised, including the phenomenological approaches of Refs.~\cite{Cardoso:2019mqo,McManus:2019ulj,Volkel:2022aca}, the eikonal approach of Refs.~\cite{Glampedakis:2019dqh,Silva:2019scu,Bryant:2021xdh} and the first-order formulation of Refs.~\cite{Langlois:2021xzq,Langlois:2021aji,Langlois:2022ulw}.

The results of our analysis can be conveniently summarized in a simple procedure for the computation of the quasinormal modes of a coupled linear system in Schr\"odinger-like form. We state it in the case of two coupled equations for simplicity. Under the conditions detailed in the paper, the procedure is as follows. Given a $2 \times 2$ matrix potential $\mathbf{V}(\tau)$, where $\tau$ is the radial tortoise coordinate, one has to:
\begin{itemize}
\item find the eigenvalues $v_i(\tau)$ of $\mathbf{V}(\tau)$ and identify the maxima of $v_i(\tau)$, determined by the condition $v_{i}'(\hat{\tau}_i)=0$. Each physical maximum will be associated with a tower of quasinormal modes.
\item for each value of $\hat{\tau}_i$, go back to the original formulation in terms of $\mathbf{V}(\tau)$ and consider the Taylor expansion of $\mathbf{V}(\tau)$ around $\hat{\tau}_i$.
\item perform a \emph{$\tau$-independent} change of basis that diagonalizes the constant term in the expansion of $\mathbf{V}(\tau)$. The potential takes now the form \eqref{eq:potential_expanded}, up to row exchange, where $x=\tau-\hat{\tau}_i$.
\item in terms of these coefficients $a_i,b_i,\dots$, the leading order quantization condition is given by eq.~\eqref{eq:leading_qnm}, the quadratic mixing correction by eq.~\eqref{eq:mixing_qnm} and the one-loop anharmonic corrections are captured by eq.~\eqref{eq:oneloop_qnm}.
\end{itemize}
These analytic formulas for the computation of the quasinormal frequencies, together with the conceptual formulation of the approximation scheme, are the main results of this work. They can be thought of as a generalization of the classic results of Schutz and Will, which apply for a non-matrix (i.e. $1 \times 1$) potential~\cite{Schutz:1985km}.
It is also worth stressing that our formalism applies also to the case of an $N \times N$ potential ${\bf V}$. The steps outlined above are applicable to the full set of eigenvalues of ${\bf V}$ and the only non-trivial question is how to generalize the computation of the mixing corrections. At the order we consider in perturbation theory, it is enough to consider separately the mixing corrections from each off diagonal term $(i,j)$ (with $i\neq j$), which can be extracted by considering the $2\times 2$ submatrix 
\begin{equation}
\begin{pmatrix}
V_{ii}(x) & V_{ij}(x) \\
V_{ji}(x)  & V_{jj}(x)\\
\end{pmatrix},
\end{equation}
and are given by the analogue of eq.~\eqref{eq:mixing_qnm}. The corrections to $\omega^2$ from all the mixing terms with $j \neq i$ should then be summed.

Even though the leading order quantization condition~\eqref{eq:leading_qnm} may look formally analogous to that derived for the $N=1$ case in Ref.~\cite{Schutz:1985km}, the result of our procedure, already at leading order, is \emph{inequivalent} to that obtained by just neglecting the mixing terms from the start. Indeed our analysis is necessary to correctly identify the point around which to expand the potential. Ignoring the mixing terms from the beginning will lead to an incorrect identification of the appropriate maximum. In addition, the constant off-diagonal terms are reabsorbed in the third step of the previous procedure: the mixing terms (not necessarily small) are partly taken into account already by the leading order result. Finally, the subleading mixing and anharmonic corrections we provide a systematic way to go beyond the leading order quantization condition.

The article is organized as follows. In Section~\ref{sec:perturbations} we introduce the Schr\"odinger-like equations governing the linear dynamics of perturbations around black holes and review the usual formulation of the boundary value problem associated with QNMs. In Section~\ref{sec:hamiltonian} we then reformulate the problem, first in terms of a classical Hamiltonian dynamical system and then in terms of a quantum mechanical system evolving with a non-hermitian time-dependent evolution operator. We use this formulation to outline our approximation scheme, based on a generalized adiabatic theorem. We then derive analytically the (approximate) quantization conditions for the QNMs for a general system of two coupled equations, in Section~\ref{sec:approximation}, including anharmonic corrections up to one-loop (corresponding to a second order WKB analysis) and mixing correction up to quadratic order in the mixing parameter. We validate our approach by studying an explicit example in Section~\ref{sec:example},  and then draw our conclusions in Section~\ref{sec:discussion}. The detailed proofs of some technical results of relevance for the paper are provided in the appendices. In particular, Appendix~\ref{app:Schreqs} contains a proof that the system of equations for the perturbations can be always put in Schr\"odinger-like form whenever an action formulation is available, and Appendix~\ref{app:adiabatic} contains a proof of the generalized adiabatic theorem for non-hermitian evolution operators acting on a finite dimensional vector space.

\section{Perturbations around black holes}\label{sec:perturbations}

In Einstein gravity the equations governing the linear dynamics of massless spin $s$ perturbations $\psi_{s}(r_\star)$ ($s=0,1,2$) around black holes can be cast in the form of a Schr\"odinger-like equation of the form 
\begin{equation}
\left( \dfrac{\rm d^{2}}{{\rm d}r_\star^{2}} + \omega^{2} - V(r_\star) \right) \psi_{s}(r_\star) =0,
\end{equation}
where $r_\star$ corresponds to the radial tortoise coordinate, the time dependence of the perturbation has been taken as $e^{-i \omega t}$ and the angular dependence has been separated through the use of the appropriate angular functions --- spherical harmonics for Schwarzschild black holes and spin-weighted spheroidal harmonics for Kerr black holes. The potential function $V(r_\star)$ --- which depends on the background, the spin $s$ of the field and the angular momentum $\ell$ of the perturbation --- can be chosen to be real and in the case of gravitational ($s=2$) perturbations on the Schwarzschild background is given by the Regge--Wheeler (odd sector) and Zerilli (even sector) potentials; with these choices, in the rotating case, it also depends on the frequency $\omega$.

More generally, in the effective field theory approach to black hole perturbations in scalar-tensor theories it is always possible to obtain a Schr\"odinger-like equation
for the perturbations around spherically symmetric~\cite{Franciolini:2018uyq} or slowly rotating~\cite{Hui:2021cpm} \emph{hairy black holes}, backgrounds with a non trivial static scalar profile. We show this explicitly and in full generality, under a stability requirement, in Appendix~\ref{app:Schreqs}. However, typically, one has to deal with a system of coupled equations rather than a single equation. This is the case, in particular, for the even sector of a spherically symmetric background, where the scalar degree of freedom is coupled to the even gravitational perturbation. 

Our goal is to develop an analytic approximation scheme to compute the quasinormal frequencies of a Schr\"odinger-like system of $N$ coupled equations in terms of a finite number of coefficients, generalizing the approach first introduced in Ref.~\cite{Schutz:1985km}.

\subsection{The quasinormal modes}

The quasinormal modes are defined by the following boundary value problem: given a smooth $N \times N$ matrix function $\mathbf{V}(r_\star)$ (the potential) defined for $-\infty < r_\star < \infty$ and satisfying the fall-off condition
\begin{equation}
\label{eq:fall_off}
\lim_{r_\star \rightarrow \pm\infty} \lvert r_\star\,\mathbf{V}(r_\star) \rvert =\mathbf{0},
\end{equation}
we want to find the values of $\omega$ such that the Schr\"odinger-like system of equations 
\begin{equation}
\label{eq:general_system}
\left( \dfrac{\rm d^{2}}{{\rm d}r_\star^{2}} + \omega^{2} - \mathbf{V}(r_\star) \right) \vec{\psi}(r_\star) =0,
\end{equation}
admits a solution satisfying the asymptotic conditions
\begin{equation}
\vec{\psi} \sim 
\begin{cases}
\vec{\psi}_{+} e^{+\mathrm{i} \omega r_\star} \qquad {\rm for} \quad r_\star\rightarrow +\infty \\
\vec{\psi}_{-} e^{-\mathrm{i} \omega r_\star} \qquad {\rm for} \quad r_\star\rightarrow -\infty . \\
\end{cases}
\end{equation}
Physically, these conditions correspond to an outgoing (right-moving) wave at spatial infinity $r_\star \rightarrow +\infty$ and to an ingoing (left-moving) wave at the black hole horizon $r_\star \rightarrow -\infty$.
The special values of $\omega$ that admit such a solution are known as the quasinormal frequencies $\omega_{\rm qnm}$ and the corresponding functions are the quasinormal modes. Due to the boundary conditions, this boundary value problem is \emph{not} self-adjoint, so that the quasinormal frequencies can be (and usually are) complex numbers and the corresponding functions are usually not normalizable.

We shall assume that the eigenvalues of the potential $\mathbf{V}(r_\star)$ are real. In some instances the potential can be a function of the frequency $\omega$ itself; in such a case we shall assume that the eigenvalues are real for real $\omega$. In what follows we shall work under the simplifying assumption of a potential independent of the frequency and then comment on how to generalize our methods to the case of frequency-dependent potentials. In Sec.~\ref{sec:example} we treat an example in which the eigenvalues of $\mathbf{V}$ are complex in a range of $r_\star$ and discuss the regime of validity of our approach. Notice that, as we show in Appendix~\ref{app:Schreqs}, it is always possible to obtain a Schr\"odinger-like system with hermitian potential for the perturbations of a stationary spacetime background if an action formulation is available. In those cases the eigenvalues of the potential will be real (for real $\omega$), and our approach will be thus applicable.
Moreover, as it will become clear in the following sections, the validity of our approach (and of that of Schutz and Will~\cite{Schutz:1985km} as well) is related to the smoothness of the potential.
Examples with perturbations, such as those of Ref.~\cite{Cheung:2021bol}, can violate this assumption, since, depending on the shape of the perturbation (\emph{e.g.} a bump), the derivative of the evolution operator can be large even for perturbations of arbitrarily small amplitude. However, as discussed in~\cite{Mirbabayi:2018mdm,Hui:2019aox}, over short time scales the ringdown is described by the QNMs of the unperturbed potential and the use of our approach is justified.

\section{The Hamiltonian formulation and a generalized adiabatic theorem}\label{sec:hamiltonian}

The system of differential equations~\eqref{eq:general_system} admits an Hamiltonian formulation. In order to make contact with the standard notation of Hamilton's mechanics, in this section we shall map the variable $r_\star$ to a fictitious time coordinate $\tau$ and the fields $\psi_{i}$ to the generalized coordinates $q_{i}$, and introduce the canonical conjugate momenta $p_{i}$. The time dependent Hamiltonian 
\begin{equation}
H(q,p,\tau) = \dfrac{1}{2} p_{i}p_{i} + \dfrac{1}{2} q_{i} \left( \omega^{2}\delta_{ij} -V_{ij}(\tau) \right) q_{j},
\end{equation}
reproduces the system of equations~\eqref{eq:general_system}, as can be easily checked by using Hamilton's equations 
\begin{equation}
\dot{q}_{i}= \dfrac{\partial H}{ \partial p_{i}} = p_{i}, \qquad \dot{p}_{i}=- \dfrac{\partial H}{ \partial q_{i}} =- \left( \omega^{2}\delta_{ij} -V_{ij}(\tau) \right) q_{j}.
\end{equation}
In this language it becomes transparent that it is not possible in general to bring the system of differential equations~\eqref{eq:general_system} in diagonal form through a point transformation $Q_{i}= f_{i}(q_{j},\tau)$. Indeed, such a change of variable corresponds to a canonical transformation with generating function $F_{2}(q_{i},P_{i})= P_{i}f_{i}(q_{j},\tau)$, so that the generalized coordinates and the Hamiltonian transform as 
\begin{equation}
p_{i} \rightarrow \dfrac{\partial F_{2}}{\partial q_{i}} = P_{j} \dfrac{\partial f_{j}(q_{k},\tau)}{\partial q_{i}}, \quad
Q_{i} \rightarrow \dfrac{\partial F_{2}}{\partial P_{i}} = f_{i}(q_{k},\tau), \quad
H(q,p) \rightarrow H(q(Q),p(P,Q))+\dfrac{\partial F_{2}}{\partial \tau}.
\end{equation}
The term $H(q(Q),p(P,Q))$ in the new Hamiltonian does not include terms linear in the generalized momenta $P_{i}$, whereas the term $\frac{\partial F_{2}}{\partial \tau}$ is linear in the momenta, being equal to $P_{i}\frac{\partial f_{i}(q_{j},\tau)}{\partial \tau}$. Therefore, the point transformation always introduces a mixing term between the new coordinates and momenta, unless $f_{i}(q_{j})$ is time independent. But a time independent coordinate transformation can diagonalize the potential term $q_{i}V_{ij}(\tau)q_{j}$ for all times $\tau$ only if the latter has the simple form $V_{ij}(\tau)=v(\tau) \bar{V}_{ij}$, with constant $\bar{V}_{ij}$, which is not the case in general.

Let us go back to the boundary problem and try to gain further insights in the Hamiltonian language. Hamilton's equations written in matrix form are
\begin{equation}
\dfrac{\rm d}{{\rm d}\tau} 
\begin{pmatrix}
q_{i} \\ p_{i} 
\end{pmatrix}
=
\begin{pmatrix}
0 & \hspace{5pt} \delta_{ij} \\
-\omega^{2} \delta_{ij} + V_{ij}(\tau) & \hspace{5pt} 0 \\
\end{pmatrix}
\begin{pmatrix}
q_{j} \\ p_{j} 
\end{pmatrix},
\end{equation}
which, when multiplied by the imaginary unit $\mathrm{i}$, is formally the time-dependent Schr\"odinger equation with a non-hermitian evolution operator that we denote with $\mathbf{K}$.
Since $\vert \tau\, V_{ij}(\tau)\vert \rightarrow 0 $ for $\tau \rightarrow \pm \infty$, it is natural to perform a (complex) canonical transformation that diagonalizes the time evolution in the asymptotic region.\footnote{The connection between a classical system of coupled harmonic oscillators, complex coordinates and quantum mechanics was pointed out long ago by Strocchi~\cite{Strocchi:1966}. Here, however, we are using this correspondence in a broader sense, trying to acquire intuition on the dynamics of a classical system of coupled oscillators by using some familiar intuition from quantum mechanics.}
Defining
\begin{equation}
\xi_{i}^{+} = \dfrac{1}{\sqrt{2}}\left( \sqrt{\omega}\,q_{i} +\mathrm{i} \dfrac{p_{i}}{\sqrt{\omega}}\right), \qquad \xi_{i}^{-} = \dfrac{1}{\sqrt{2}}\left( \sqrt{\omega}\,q_{i} - \mathrm{i} \dfrac{p_{i}}{\sqrt{\omega}}\right), 
\end{equation}
the evolution equation becomes 
\begin{equation}
\mathrm{i}\dfrac{\rm d}{{\rm d}\tau} 
\begin{pmatrix}
\xi_{i}^{+} \\ \xi_{i}^{-}
\end{pmatrix}
=
\left[ 
\begin{pmatrix}
+\omega & 0 \\
0 & -\omega \\
\end{pmatrix} \delta_{ij}
- \dfrac{V_{ij}(\tau)}{2\omega}
\begin{pmatrix}
+1 & +1 \\
-1 & -1\\
\end{pmatrix}
\right]
\begin{pmatrix}
\xi_{j}^{+} \\ \xi_{j}^{-}
\end{pmatrix},
\end{equation}
while the boundary conditions are simply
\begin{equation}
\label{eq:boundary_conditions}
\begin{cases}
&\xi_{i}^{+} \rightarrow 0, \quad \hspace{16pt} \xi_{i}^{-} \sim e^{\mathrm{i} \omega \tau}, \qquad {\rm for} \quad \tau \rightarrow +\infty \\
&\xi_{i}^{+} \sim e^{-\mathrm{i} \omega \tau}, \quad \xi_{i}^{-} \rightarrow 0, \qquad \hspace{10pt}{\rm for} \quad \tau \rightarrow -\infty. 
\end{cases}
\end{equation}
In this basis it becomes transparent that for $\omega$ real in the asymptotic region $\vert \tau \vert \rightarrow \infty$ the evolution is unitary and governed by the asymptotic eigenvalues $\pm \omega$. The boundary conditions correspond to the requirement that the system undergoes a transition from an eigenstate with eigenvalue $+\omega$ at early times to one with eigenvalue $-\omega$ at late times. In the case of a slowly varying Hermitian evolution operator, the adiabatic theorem of quantum mechanics implies that transitions occur efficiently only when two eigenvalues become (almost) degenerate, as in the well-known Landau-Zener effect (\emph{i.e.}~close to a level crossing, usually lifted by perturbations in a quantum mechanical system).

We shall establish a mild generalization of the adiabatic theorem for evolution operators that are diagonalizable, have a real spectrum, and act on a finite dimensional vector space. 

\subsection{An adiabatic theorem}
\label{sec:adiabatic}

The evolution operator $\mathbf{K}$ can be regarded as an operator acting on the $2N$-dimensional (complex) vector space with elements $(\mathbf{q},\mathbf{p})^{T}$. 
Let us assume that the potential $\mathbf{V}(\tau)$, acting on the $N$-dimensional $\mathbf{q}$ subspace, is a diagonalizable operator with real eigenvalues~$v_{i}$, $i=1,\dots,N$.\footnote{We often suppress the explicit time dependence for ease of notation.} Then the evolution operator
\begin{equation}
\mathbf{K}(\tau)= \begin{pmatrix}
0 & \hspace{5pt} \mathrm{i}\,\mathbf{1}_{N} \\
-\mathrm{i}\,\omega^{2} \mathbf{1}_{N} + \mathrm{i}\, \mathbf{V}(\tau) & \hspace{5pt} 0 \\
\end{pmatrix}
\end{equation}
is diagonalizable for every time $\tau$, with eigenvalues determined by the zeros of the characteristic polynomial:
\begin{equation}
\label{eq:char_pol}
p(\lambda)=\det\left( \lambda \mathbf{1}_{2N} -\mathbf{K}\right)= \det\left( \lambda^{2} \mathbf{1}_{N}- \omega^{2} \mathbf{1}_{N} + \mathbf{V}(\tau)\right),
\end{equation}
where in the second equality we used that the identity matrix commutes with every operator. From eq.~\eqref{eq:char_pol} it follows that the eigenvalues occur always in pairs and are such that their square is an eigenvalue of $(\omega^{2} \mathbf{1}_{N} - \mathbf{V})$. In particular:
\begin{equation}
\label{eq:eigenvalues}
\lambda_{i,\pm}= \pm \sqrt{\omega^{2}-v_{i}}.
\end{equation}

Since the operator $\mathbf{K}$ is not Hermitian its eigenvectors $\lvert k_{a} \rangle$ (with $a=1,\dots,2N$) do not form an orthonormal basis with respect to the standard inner product; however, \emph{they still form a basis}. Moreover, since the Hilbert space is finite dimensional we can consider the dual basis $\langle \tilde{k}_{a} \rvert$, with $a=1,\dots,2N$, defined by
\begin{equation}
\langle \tilde{k}_{a} \vert k_{b} \rangle = \delta_{ab},
\end{equation} which exists and is unique. The dual basis vectors act as projectors on the desired eigenspaces. By the uniqueness of the dual basis, it follows that given an eigenvector $\lvert k_{a} \rangle$ of $\mathbf{K}$ with eigenvalue $\lambda_{a}$, the dual vector $\langle \tilde{k}_{a} \rvert$ is a left eigenvector with eigenvalue $\lambda_{a}$:
\begin{equation}
\label{eq:left_eig}
\langle \tilde{k}_{a} \rvert \mathbf{K}= \lambda_{a} \langle \tilde{k}_{a} \rvert .
\end{equation}

As a last ingredient for the proof of the adiabatic theorem we notice that for any fixed value of $\omega^{2}>0$, at early times $\tau\rightarrow -\infty$ the eigenvalues $\lambda_{i,\pm}$ are real, for every $i$. This follows simply from the eigenvalue equation~\eqref{eq:eigenvalues}, the assumption that $\omega$ and $v_{i}(\tau)$ are real and the fall-off condition \eqref{eq:fall_off}, which implies $v_{i}(\tau\rightarrow -\infty)\rightarrow 0$. By continuity, the reality of the eigenvalues can be violated only after a level crossing: as $v_{i}(\tau) \rightarrow \omega^{2}$ the two eigenvalues $\lambda_{i,\pm}\rightarrow 0$ and become degenerate.

As we discuss in detail in Appendix~\ref{app:adiabatic}, these preliminary considerations ensure that the standard proof of the adiabatic theorem of quantum mechanics goes through for this system, even though the evolution operator $\mathbf{K}(\tau)$ is not Hermitian. As a result, away from level crossing and for smooth potentials, level transitions are suppressed. 

In our argument this is the only step where the assumption of real $\omega$ comes into play: the generalized adiabatic theorem, valid for $\omega$ real, allows us to single out level crossing as special points to focus on.

\subsection{A heuristic argument on level crossing and quasinormal modes}

The boundary conditions for the quasinormal modes~\eqref{eq:boundary_conditions} require a level transition from an eigenstate with eigenvalue $+\omega$ to one with with eigenvalue $-\omega$.
From the generalized adiabatic theorem of the previous section it follows that for well-behaved potentials, with slow variation and real eigenvalues, such a transition occurs efficiently only if a pair of eigenvalues $\lambda_{i,\pm}(\tau)$ (one positive and one negative) becomes degenerate, by touching or crossing $0$ at a given time $\hat{\tau}_i$. Thinking about a fixed value of $i$ for definiteness, the condition for this to happen is:
\begin{equation}
\lambda_{i,\pm}^{2}(\hat{\tau}_i)=0 \quad \implies \quad \omega^{2}-v_{i}(\hat{\tau}_i)=0.
\end{equation}
Starting from $\lambda_{i,\pm}^{2}(\tau)>0$ at $\tau\rightarrow -\infty$ the eigenvalues approach $0$ as $\tau \rightarrow \hat{\tau}_i$. Assuming the $v_{i}(\tau)$ is a smooth function of $\tau$ this can happen either with $v_{i}'(\hat{\tau}_i)=0$ or $v_{i}'(\hat{\tau}_i)\neq 0$. 

To gain some intuition we can consider the one-dimensional case $N=1$. In this case $v_{i}(\tau) \equiv V(\tau)$ and the condition $V'(\hat{\tau})=0$ corresponds to the well-know condition that the real part of the quasinormal frequencies $\omega^{2}$ is equal to the maximum of the potential, familiar from the approximation scheme of~\cite{Schutz:1985km}. Heuristically this can be justified by noticing that the quasinormal boundary conditions can be formulated requiring the outgoing waves to have equal amplitudes both at the horizon and at spatial infinity. In the analysis of~\cite{Schutz:1985km}, based on matching the ingoing and outgoing WKB asymptotics, this condition can not be accommodated if the frequency $\omega^2$ (or its real part) is smaller than the maximum of the potential. Indeed in this scenario there would be a region with exponential damping, making it impossible to satisfy the boundary conditions. 
We shall make the assumption that also in the multidimensional case, $N>1$, only the crossing points with $v_{i}'(\hat{\tau}_i)=0$ are relevant for the computation of the quasinormal frequencies.

For well-behaved physical potentials, the total number of real maxima in the physical coordinate region is expected to be $N$, the same as the dimension of the system of linear equations. In this case, every maximum $v_{i}(\hat{\tau}_i)$ with $i=1,\dots, N$ will give rise to a tower of quasinormal modes.

Up to now we have been assuming that $\omega$ is real and that the eigenvalues $v_{i}(\tau)$ are real. The former should be understood as an approximation, under the condition that the imaginary part of $\omega$ is much smaller than the real part. The latter condition can be relaxed in physically realistic examples, again on a heuristic basis.

\section{An approximation scheme for computing quasinormal frequencies}
\label{sec:approximation}

Having determined the interesting points $\hat{\tau}_i$ as the values of $\tau$ that correspond to stationary points of the eigenvalues of the potential matrix, defined by the condition $v_{i}'(\hat{\tau}_i)=0$, we shall focus on the dynamics of the system by going back to the original formulation~\eqref{eq:general_system} and expanding the potential in a power series around every maximum $v_{i}(\hat{\tau}_i)$. Moreover, for notational convenience we change independent variable from $\tau$ to $x=\tau-\hat{\tau}_i$, for every fixed value of $i$.\footnote{The variable $x$ is a continuous variable that implicitly depends on the (discrete) index $i$. We keep this dependence implicit for ease of notation.}

To illustrate the procedure let us focus on the case $N=2$ and choose $i=1$ for concreteness. Consider the Taylor expansion of the matrix $\mathbf{V}(x)$ around $x=0$, corresponding to a maximum of $v_{1}(x)$:
\begin{equation}
\mathbf{V}(x) = \sum_{k=0}^{\infty} \mathbf{V}^{(k)} x^{k}.
\end{equation} 
As a first step we bring the constant matrix $\mathbf{V}^{(0)}$ to diagonal form by performing a constant change of variables.\footnote{We diagonalize the matrix over the complex field. In principle it could be possible to encounter a matrix which is non-diagonalizable. However the set of $N\times N$ non-diagonalizable matrices has zero Lebesgue measure over $\mathbb{C}^{N\times N}$ so this is not expected to be a problem in practical applications, as we shall see in explicit examples.} In this basis, the potential takes now the form
\begin{equation}\label{eq:potential_expanded}
\mathbf{V}(x) = \begin{pmatrix}
a_{0}-a_{2} x^{2} + \dots & \hspace{5pt} b_{1} x + \dots\\
c_{1} x + \dots & d_{0} + \dots \\
\end{pmatrix},
\end{equation}
where the absence of the term $a_{1}$ follows from the fact that we are expanding around a local maximum of $v_{1}(x)$, while $b_{0}$ and $c_{0}$ have been reabsorbed by the constant diagonalization. The coefficients $d_1,b_2,c_2,\dots$ are in general all non-zero, but will give subleading corrections to the quasinormal frequencies, as we shall explain in the following. We shall denote the first diagonal component of $\mathbf{V}(x)$ in this particular basis as $V_{11}(x)$.

In a situation where the mixing terms $b_{k}, c_{k}$ (with $k\geq 1$) are small we can adopt a perturbative approach and introduce a formal counting parameter $\varepsilon$. This is the case, for instance, when they arise from a weakly coupled extension of GR, in which they are proportional to the coupling constants of the underlying theory. Notice, however, that the effect of $b_{0},c_{0}$ has been fully captured in our approach by diagonalizing the constant potential matrix at the maximum; such an effect is not necessarily small, so that our approach can allow to compute the quasinormal modes also in scenarios in which the mixing terms are of order one, but the residual coupling after diagonalization of the constant matrix are small as we shall see in the next section.

The leading order approximation for the tower of quasinormal frequencies associated with $v_{1}(\hat{\tau}_1)$, is obtained by neglecting the mixing terms and considering only the dynamics of $q_{1}(x)$, as governed by $V_{11}(x)$. At lowest order, the behaviour of the quasinormal modes can then be captured by the analysis of~\cite{Schutz:1985km}, considering only the quadratic approximation for the potential $V_{11}(x) \simeq a_{0} - a_{2} x^{2}$. Anharmonic corrections to $V_{11}(x)$ can be included by either performing a higher order WKB analysis of the asymptotics~\cite{Iyer:1986np,Matyjasek:2017psv,Konoplya:2019hlu} or by mapping the question to a bound state problem~\cite{Blome:1984,Ferrari:1984ozr,Ferrari:1984zz,Zaslavsky:1991ug,Hatsuda:2019eoj}. In order to include also the corrections deriving from the mixing terms it proves useful to adopt the latter approach, which allows both for a conceptually clear formulation of the computation and a straightforward method of calculation. We shall adopt a perturbative approach and compute separately corrections of two classes: mixing terms among $q_{1}$ and $q_{2}$, and anharmonic terms in the potential. 

\subsection{Mixing terms}
\label{sec:mixing}

Let us consider without loss of generality the system of equations obtained by expanding around the maximum of $v_{1}(x)$, and keep only terms that are formally of order $x^{2}$ or lower with respect to $V_{11}(x)$. That is: if we formally assign units to $x$, we only keep terms whose contribution to the tower of quasinormal frequencies associated with $V_{11}$ (as defined in the limit $\varepsilon\rightarrow 0$) is of the same order as the contribution generated by $a_2$. As it will become clear (see eq.~\eqref{eq:mixing_qnm}), mixing contributions of this order are generated only by $b_1,c_1,d_0$.
Introducing explicitly the counting parameter $\varepsilon$, we have the system
\begin{equation}\label{eq:quad_mix}
\dfrac{\rm d^{2}}{{\rm d}x^{2}} 
\begin{pmatrix}
q_{1} \\ q_{2} 
\end{pmatrix}
=
-\omega^{2}\begin{pmatrix}
q_{1} \\ q_{2} 
\end{pmatrix}+\begin{pmatrix}
a_{0}-a_{2} x^{2} & \hspace{5pt} \varepsilon b_{1} x \\
\varepsilon c_{1} x & d_{0} \\
\end{pmatrix}
\begin{pmatrix}
q_{1} \\ q_{2} 
\end{pmatrix},
\end{equation}
with boundary conditions such that $q_{i}$ is left-moving at the horizon and right-moving at infinity, as in~\eqref{eq:boundary_conditions}. Since we are studying the set of quasinormal modes associated with $v_{1}(\hat{\tau}_1)$, we are interested in the eigenmodes that reduce to the $q_{1}$ eigenfunction in the limit $\varepsilon\rightarrow 0$. To understand the asymptotic behaviour of the system we can use the ordinary WKB approximation for a system of coupled differential equations~\cite{Weinberg:1962eikonalmhd}, considering the limit $x\rightarrow \pm \infty$.\footnote{By ordinary WKB approximation, we refer to the method developed for a single time-independent Schr\"odinger equation with slowly varying potential by Wentzel, Kramers and Brillouin (see for instance~\cite{Sakurai:2011zz} for a textbook treatment), and its generalization to systems of equations developed by Weinberg~\cite{Weinberg:1962eikonalmhd}. This should not be confused with the WKB-inspired approach of Schutz and Will~\cite{Schutz:1985km}.} At leading order in the WKB expansion we look for solutions of the form
\begin{equation}
\begin{pmatrix}
q_{1} \\ q_{2} 
\end{pmatrix} =
\begin{pmatrix}
A \\ B
\end{pmatrix}
e^{\I \,S(x)},
\end{equation}
with $A=B=\rm const$ and $S''(x) \ll (S'(x))^{2}$. Neglecting terms proportional to $S''(x)$, our system of coupled differential equations becomes an algebraic system in $u\equiv S'(x)$:
\begin{equation}
\begin{pmatrix}
u^{2}+\omega^{2}+a_{0}-a_{2} x^{2} & \hspace{5pt} \varepsilon b_{1} x \\
\varepsilon c_{1} x & u^{2}+\omega^{2}+d_{0} \\
\end{pmatrix}
\begin{pmatrix}
A \\ B
\end{pmatrix}
= 0.
\end{equation}
The system admits non-trivial solutions when it becomes degenerate; the vanishing of the determinant gives the conditions
\begin{equation}
\begin{split}
&u_{\pm}^{2} = \dfrac{1}{2} \left(2\omega^{2}-a_{0}+a_{2}x^{2}-d_{0}\pm \mathcal{C}\right),\\
&\mathcal{C}=\sqrt{(2\omega^{2}-a_{0}+a_{2}x^{2}-d_{0})^{2}-4(\omega^{2}-a_{0}+a_{2}x^{2})(\omega^{2}-d_{0})-\varepsilon^{2} b_{1}c_{1} x^{2}}.
\end{split}
\end{equation}
In the limit $x\rightarrow \pm \infty $ the two WKB solutions then become:
\begin{equation}
\begin{split}
& u_{+} \sim \sqrt{a_{2}} x \quad \implies \quad \vec{q}_{+} \sim \begin{pmatrix}
A_{+} \\ B_{+}
\end{pmatrix} 
e^{\pm \I \, \int^{x} \sqrt{a_{2}}x {\rm d}x} \sim \begin{pmatrix}
A_{+} \\ B_{+}
\end{pmatrix} 
e^{\pm \I \, \frac{\sqrt{a_{2}}}{2}x^{2}} , \\
& u_{-} \sim {\rm const} \quad \implies \quad \vec{q}_{-} \sim \begin{pmatrix}
A_{-} \\ B_{-}
\end{pmatrix} 
e^{\pm \I \, ({\rm const}) \,x }.
\end{split}
\end{equation}
From the asymptotic analysis we learn that the eigenfunctions that satisfy the quasinormal mode boundary conditions and reduce to the $q_1$ eigenfunction in the limit $\varepsilon \rightarrow 0$ are those associated with $u_{+}$:
\begin{equation}
\vec{q}_{+} \sim \begin{pmatrix}
A_{+} \\ B_{+}
\end{pmatrix} 
\exp\left(+ \I \, \dfrac{\sqrt{a_{2}}}{2}x^{2}\right) \qquad \qquad x\rightarrow \pm \infty .
\end{equation}
Moreover, the asymptotic behaviour of the solutions suggests performing the analytic continuation
\begin{equation}\label{eq:analytic_continuation_1}
x \rightarrow e^{+\I \, \pi/4} z, \qquad x^{2} \rightarrow + \I \,z^{2}, \qquad \dfrac{\rm d}{{\rm d}x^{2}} \rightarrow -\I \, \dfrac{\rm d}{{\rm d}z^{2}},
\end{equation}
so that the eigenvalue problem~\eqref{eq:quad_mix} is mapped to a new eigenvalue problem with \emph{regular} boundary conditions:
\begin{equation}
\vec{q}(z) \rightarrow 0 \qquad z \rightarrow \pm \infty.
\end{equation}
We shall work under the assumption that the analytic continuation can be performed without encountering singularities in the complex plane, so that the procedure allows to recover the quasinormal modes from the eigenvalues of the associated bound state problem. A similar (but distinct) procedure has been discussed in the case of a single differential equation in Schr\"odinger-like form in Refs.~\cite{Blome:1984,Ferrari:1984ozr,Ferrari:1984zz,Zaslavsky:1991ug}.

Multiplying by $-\I /2$ and defining
\begin{equation}
\label{eq:redefinitions}
\nu \equiv \dfrac{\I}{2} \left(\omega^{2} -a_{0}\right), \qquad \tilde{b}_{1} \equiv \dfrac{e^{+\I \, 3\pi/4}}{2} b_{1}, \qquad \tilde{c}_{1} \equiv \dfrac{e^{+\I \, 3\pi/4}}{2} c_{1}, \qquad \tilde{d}_{0} \equiv \dfrac{\I}{2}(d_{0}-a_{0}),
\end{equation}
we obtain the system
\begin{equation}
\label{eq:schrodinger}
-\dfrac{1}{2} \dfrac{\rm d^{2}}{{\rm d}z^{2}} 
\begin{pmatrix}
q_{1} \\ q_{2} 
\end{pmatrix}
+
\begin{pmatrix}
\frac{1}{2} a_{2} z^{2} & \hspace{5pt} \varepsilon \tilde{b}_{1} z \\
\varepsilon \tilde{c}_{1} z & \tilde{d}_{0} \\
\end{pmatrix}
\begin{pmatrix}
q_{1} \\ q_{2} 
\end{pmatrix}
=
\nu\begin{pmatrix}
q_{1} \\ q_{2} 
\end{pmatrix}.
\end{equation}
Working perturbatively in $\varepsilon$ we first of all recover the leading order formula of Schutz and Will~\cite{Schutz:1985km}. Indeed, working at zeroth order in $\varepsilon$, the eigenstate with $q_{2}^{(0)}=0$ satisfies the Schr\"odinger equation and the boundary conditions for the bound states of a quantum mechanical harmonic oscillator.\footnote{From now on, a superscript $(k)$ denotes a quantity of $k$-th order in $\varepsilon$. Moreover we shall adopt the notation $q_{i,n}^{(0)}$, with $i=1,\dots,N$ and $n\in \mathbb{N}$, to denote the $i$-th component of the eigenvector associated with the $n$-th bound state of our eigenvalue problem.} The eigenfunctions and eigenvalues are then easily obtained:
\begin{equation}
\label{eq:leading_qnm_0}
\begin{split}
&\nu_{n}^{(0)} = \sqrt{a_{2}} \left(n+\dfrac{1}{2}\right) \\
& q_{1,n}^{(0)}(z) = Z_{n}\, H_{n}(\sqrt[4]{a_{2}}\, z) \exp\left(- \sqrt{a_{2}} z^{2}/2 \right), \qquad q_{2,n}^{(0)}(z)=0,
\end{split}
\end{equation}
where $H_{n}(\cdot)$ are Hermite polynomials and we normalize the eigenfunction to have unit $L^{2}$ norm, with
\begin{equation}
\label{eq:normalization}
Z_{n}= \dfrac{1}{\sqrt{2^{n} n!}} \left(\dfrac{\sqrt{a_{2}}}{\pi} \right)^{\frac{1}{4}}.
\end{equation}
Expressing the eigenvalues in terms of the original variables we recover the quantization condition for the quasinormal modes of~\cite{Schutz:1985km}, without performing a WKB matching analysis:
\begin{equation}\label{eq:leading_qnm}
\omega_{n}^{2} = a_{0} - \I \, \sqrt{4 a_{2}} \left(n+\dfrac{1}{2}\right).
\end{equation}
In order to compute the leading $\varepsilon$-corrections to the quantization condition we use the well-known Rayleigh--Schr\"odinger perturbation theory approach, familiar from the time-independent eigenvalue problems of ordinary quantum mechanics. Since the perturbation is off-diagonal, the first non-trivial correction arises only at second order in perturbation theory.

We denote by $\delta\mathcal{H}_{\varepsilon}$ the order $\varepsilon$ perturbation in equation~\eqref{eq:schrodinger}:
\begin{equation}
\delta\mathcal{H}_{\varepsilon} = 
\begin{pmatrix}
0 & \hspace{5pt} \varepsilon \tilde{b}_{1} z \\
\varepsilon \tilde{c}_{1} z & 0 \\
\end{pmatrix}.
\end{equation}
Going through the usual derivation, the order $\varepsilon^2$ correction to the eigenvalue $\nu_{n}$ can be expressed as 
\begin{equation}
\label{eq:second_pt}
\nu_{n}^{(2)} = \int_{-\infty}^{+\infty} {\rm d}z \left( \vec{q}_{n}^{\,(0)}(z) \right)^{\dagger} \delta\mathcal{H}_{\varepsilon}\, \vec{q}_{n}^{\,(1)}(z).
\end{equation}
In the ordinary treatment given in quantum mechanics textbooks it is customary to express $\vec{q}_{n}^{\,(1)}(z)$ as a linear superposition of zeroth-order eigenfunctions $\vec{q}_{m}^{\,(0)}(z)$. In our case it will be convenient to solve directly the differential equation for the first-order eigenfunction $\vec{q}_{n}^{\,(1)}(z)$. Since $q_{2,n}^{(0)}(z)=0$ and the perturbation $\delta\mathcal{H}_{\varepsilon}$ is off-diagonal we shall only need $q_{2,n}^{(1)}(z)$. Working at first order in $\varepsilon$, the differential equation for $q_{2,n}^{(1)}(z)$ is:
\begin{equation}
\left( \dfrac{\rm d^{2}}{{\rm d}z^{2}} - \alpha_{n}^{2} \right) q_{2,n}^{(1)}(z) = 2 \varepsilon \tilde{c}_{1} z\, q_{1,n}^{(0)}(z),
\end{equation}
where we have defined $\alpha_{n}$ so that its real part is positive and 
\begin{equation}
\alpha_{n}^{2} = 2 \left(\tilde{d}_{0}- \nu_{n}^{(0)} \right) = -\sqrt{4 a_{2}} \left(n+\dfrac{1}{2}\right)+\I (d_{0}-a_{0}).
\end{equation}
We can treat the term on the right hand side as a source term and solve the differential equation using the Green's function:
\begin{equation}
\left( \dfrac{\rm d^{2}}{{\rm d}z^{2}} - \alpha_{n}^{2} \right) G_{n}(z,z') = \delta(z-z') \quad \implies \quad q_{2,n}^{(1)}(z) = 2 \varepsilon \tilde{c}_{1} \int_{-\infty}^{+\infty} {\rm d}z' G_{n}(z,z') \, z'\, q_{1,n}^{(0)}(z').
\end{equation}
The Green's function for the differential operator $\left(\frac{\rm d^{2}}{{\rm d}z^{2}} - \alpha_{n}^{2}\right)$ satisfying the appropriate boundary conditions is given by:
\begin{equation}
G_{n}(z,z') = - \dfrac{1}{2\alpha_{n}} \left[ \Theta(z'-z) e^{-\alpha_{n} (z'-z)} + \Theta(z-z') e^{-\alpha_{n} (z-z')} \right],
\end{equation}
from which it follows that
\begin{equation}
q_{2,n}^{(1)}(z) = - \dfrac{\varepsilon \tilde{c}_{1}}{\alpha_{n}} \int_{0}^{+\infty} {\rm d}\zeta \left[ (z+ \zeta) q_{1,n}^{(0)}(z+ \zeta) + (z- \zeta) q_{1,n}^{(0)}(z - \zeta) \right] e^{-\alpha_{n}\zeta} .
\end{equation}
Plugging back in eq.~\eqref{eq:second_pt} and using eq.~\eqref{eq:leading_qnm_0}, after straightforward manipulations we obtain
\begin{equation}
\nu_{n}^{(2)}= - \varepsilon^{2} \dfrac{4 \sqrt{\pi} Z_{n}^{2} \tilde{b}_{1}\tilde{c}_{1}}{\alpha_{n} a_{2}} \int_{0}^{+\infty} P_{n}(\sigma) \exp\left(-2 \frac{\alpha_{n}}{\sqrt[4]{a_{2}}}\,\sigma-\sigma^{2}\right) {\rm d}\sigma ,
\end{equation}
where $P_{n}(\sigma)$ is a polynomial of degree $(2n+2)$ with only even terms, defined by the integral
\begin{equation}
\label{eq:polynomial_mix}
P_{n}(\sigma) = \dfrac{1}{\sqrt{\pi}} \int_{-\infty}^{+\infty} (\rho-\sigma)(\rho+\sigma) H_{n}(\rho-\sigma) H_{n}(\rho+\sigma) e^{-\rho^{2}} {\rm d}\rho .
\end{equation}
Using the relations~\eqref{eq:redefinitions} and~\eqref{eq:normalization} the $\varepsilon^{2}$ correction to the quasinormal modes can be expressed as
\begin{equation}
\label{eq:mixing_qnm}
\left(\delta\omega^{2}_{n}\right)^{(2)} = - \varepsilon^{2}\,\dfrac{b_{1}c_{1}}{a_{2}} F_{n}\left(\dfrac{\alpha_{n}}{\sqrt[4]{a_{2}}}\right),
\end{equation}
where the function $F_{n}(y)$ defined by
\begin{equation}
\label{eq:function_mix}
F_{n}(y)= \left(\dfrac{1}{2^{n-1}\cdot n!}\right) \dfrac{1}{y} \int_{0}^{+\infty} P_{n}(\sigma) e^{-\sigma^{2}} e^{-2 y \sigma} {\rm d}\sigma,
\end{equation}
is known analytically and tabulated for the first few values of $n$ in Table~\ref{tab:mixing} together with the polynomials $P_{n}(\sigma)$. An asymptotic approximation for $F_{n}(y)$ can be obtained in the limit $|y| \rightarrow \infty$, as detailed in Appendix~\ref{app:asymptotic}. Moreover, it is possible to check that $\vert F_{0}(y)\vert < 1 $ for $\mathfrak{Re}(y) >0$, so that the mixing correction to the fundamental mode is always bounded by $\vert \varepsilon^2 b_1 c_1 /a_2 \vert $ and the corresponding perturbative expansion is under control as long as this parameter is small. 
In particular, the approximation scheme holds also for cases where $\varepsilon\sim \mathcal{O}(1)$, provided that $\vert b_1 c_1 /a_2 \vert $ is small, which depends on the potential.

\begin{table}[t]
\centering
\begin{tabular}{c|c|c}
$n$ & $P_{n}(\sigma)$ & $F_{n}(y)$ \\
\cline{1-3} 
\rule{0pt}{4ex}$0$ & $\left(-\sigma^{2} + \frac{1}{2}\right)$ & 
$1-\sqrt{\pi} \,e^{y^{2}} y \, {\rm erfc}(y)$ \\
\rule{0pt}{4ex}$1$ & $\left(4\,\sigma^{4}-4\, \sigma^{2} + 3 \right)$ & 
$-2 y ^2+\frac{1}{y}2 \sqrt{\pi } e^{y ^2} \left(y ^2+1\right)^2 \text{erfc}(y )-3$ \\ 
\rule{0pt}{4ex}$2$ & $4\left(-4\,\sigma^{6}+10\,\sigma^{4}-10\, \sigma^{2} + 5 \right)$ & 
$2 y ^4+9 y ^2-\frac{1}{2} \sqrt{\pi } e^{y ^2} \left(2 y ^2+5\right)^2 y \, \text{erfc}(y )+9$ \\ 
\end{tabular}
\caption{\it Explicit form of the polynomials $P_{n}$ and functions $F_{n}$ defined respectively in eqs.~\eqref{eq:polynomial_mix} and~\eqref{eq:function_mix}, relevant for the computation of the mixing correction~\eqref{eq:mixing_qnm}. The complementary error function is defined as: ${\rm erfc}(y) \equiv 2\int_y^\infty e^{-u^2} {\rm d}u / \sqrt{\pi}$.}
\label{tab:mixing}
\end{table}

\subsection{Anharmonic corrections}
\label{sec:anharmonic}

Higher order corrections to the quadratic potential can be included systematically and the effect of these terms can be captured straightforwardly in a perturbative approach. In this spirit, in this section we shall consider only the anharmonic corrections to the leading result~\eqref{eq:leading_qnm}, neglecting the mixing terms analyzed in the previous section.
Following the procedure outlined at the beginning of Section~\ref{sec:approximation} we focus on the single differential equation 
\begin{equation}
\left( \dfrac{\rm d^{2}}{{\rm d}x^{2}} + \omega^{2} - V_{11}(x) \right) q_{1}(x) =0.
\end{equation}
We consider the Taylor expansion of $V_{11}(x)$ around $x=0$, using for convenience the following notation, consistent with the previous sections:
\begin{equation}
V_{11}(x)= a_{0}- \sum_{k=2}^{\infty} a_{k} x^{k}.
\end{equation}
Since we are expanding around a maximum, the coefficient $a_{2}$ is positive in our sign conventions.

The asymptotic WKB analysis of the previous section and the analytic continuation we performed relied on the quadratic form of the potential, so the latter is no longer warranted in this case. However, as recently pointed out by Hatsuda~\cite{Hatsuda:2019eoj}, it proves convenient to perform a different analytic continuation which allows us to include efficiently anharmonic corrections up to very high orders. Following~\cite{Hatsuda:2019eoj} we introduce a deformation parameter $\hbar$ such that:
\begin{equation}
\left( - \hbar^{2} \dfrac{\rm d^{2}}{{\rm d}x^{2}} + \omega^{2} - V_{11}(x) \right) q_{1}(x) =0.
\end{equation} 
For the value $\hbar=\I$ we recover our original problem, while $\hbar$ has \emph{formally} the same role as Planck's constant in the quantum mechanical Schr\"odinger equation for a particle in the inverted potential $-V_{11}(x)$. Defining $\tilde{x} = x/ \sqrt{\hbar}$, multiplying by $1/2$ and rearranging the terms, the equation becomes
\begin{equation}
\label{eq:boundstate_anharmonic}
\left(-\dfrac{1}{2} \dfrac{\rm d}{{\rm d}\tilde{x}^{2}} + \dfrac{a_{2}}{2} \tilde{x}^{2} + \dfrac{1}{2} \sum_{k=3}^{\infty} \hbar^{k/2-1} a_{k} \, \tilde{x}^{k} \right) q_{1}(\tilde{x}) = \left( \dfrac{-\omega^{2}+a_{0}}{2 \hbar} \right) q_{1}(\tilde{x}),
\end{equation}
which makes it manifest that $\hbar$ acts as a loop counting parameter for anharmonic interactions.
The quasinormal modes are obtained as the analytic continuation at $\hbar= \I$ of the eigenvalues for the bound state problem~\eqref{eq:boundstate_anharmonic}.\footnote{It is easy to check that for the quadratic potential of Section~\ref{sec:mixing} the quantization conditions derived with this analytic continuation are equivalent to those obtained with the analytic continuation of eq.~\eqref{eq:analytic_continuation_1}.}

By dimensional analysis, the anharmonic corrections to the leading order eigenvalues will be a function of the dimensionless coupling constants
\begin{equation}
\mathfrak{a}_{k}=\dfrac{a_{k}}{(a_{2})^{(k+2)/4}}.
\end{equation}
We shall assume here that the dimensionless coefficients are small so that the perturbative expansion is under control, which is a valid assumption in the explicit examples we consider.

A straightforward method to compute corrections to the eigenvalues from anharmonic terms would be to use Rayleigh--Schr\"odinger perturbation theory. A more efficient approach was devised long ago by Bender and Wu~\cite{Bender:1969si}, as recently emphasized in this context in Ref.~\cite{Hatsuda:2019eoj}. This amounts to a systematic expansion in $\hbar$ --- or equivalently in the number of loops --- and reproduces the quantization conditions obtained through the approach of~\cite{Schutz:1985km,Iyer:1986np}, which is also an expansion in $\hbar$ (although with a different formulation). The general formulation for a bound state problem of the form~\eqref{eq:boundstate_anharmonic} and the resulting recurrence relations have been discussed in detail in Ref.~\cite{Sulejmanpasic:2016fwr}. The 1-loop quantization condition, corresponding to a second order WKB analysis, is given by~\cite{Iyer:1986np,Hatsuda:2019eoj}
\begin{equation}\label{eq:oneloop_qnm}
\nu_{n}^{(\rm 1-loop)} =\left(\delta\omega_{n}^{2}\right)^{(\rm 1-loop)} = -\dfrac{a_{3}^{2}}{32\, a_{2}^{2}}\left[7+ 60 \left(n+\dfrac{1}{2}\right)^{2}\right] + \dfrac{3\,a_{4}}{8\,a_{2}}\left[1+ 4 \left(n+\dfrac{1}{2}\right)^{2}\right].
\end{equation}
Comparing with the leading order quantization~\eqref{eq:leading_qnm_0} we see that the relative correction $\nu_{n}^{(\rm 1-loop)} / \nu_{n}^{(0)} $ depends on the expansion parameters $(\mathfrak{a}_{3},\mathfrak{a}_{4})$ as expected. We also notice that the relative correction grows linearly with $n$ in the limit of large $n$, so that the approximation is under control only for $n$ small, in a range controlled by the dimensionless couplings $(\mathfrak{a}_{3},\mathfrak{a}_{4})$. Corrections from quintic or higher interactions enter only at higher loop (or WKB) order.

\section{Explicit example}\label{sec:example}

In general, the linearized dynamics of the perturbations around a given black hole solution, described by the system \eqref{eq:general_system}, cannot be recast in the form of a set of fully decoupled Schr\"odinger-like equations. This is the case, for instance, of the perturbations of charged Kerr--Newman black holes (see for instance~\cite{Chandrasekhar:1985kt}, section 111), of massive and partially massless spinning fields on Schwarzschild or Kerr spacetimes \cite{Brito:2013yxa,Brito:2013wya,Rosen:2020crj},\footnote{Partially massless fields are special irreducible representations of massive spinning particles on Einstein spacetimes
\cite{Deser:1983tm,Deser:1983mm,Higuchi:1986py}
(the case of more general spacetimes has been discussed e.g.~in \cite{Bernard:2017tcg}). They exist for particles with spin greater than 1 when their masses take very particular values. Their main property is that they carry one degree of freedom less than generic massive fields with the same spin, thanks to a residual gauge symmetry that is responsible for removing the helicity-zero component.} and of perturbations of hairy black holes beyond GR, such as in scalar-tensor theories where the scalar has a nontrivial background profile (see, e.g., Refs.~\cite{Blazquez-Salcedo:2017txk,Kobayashi:2014wsa,Franciolini:2018uyq, Franciolini:2018aad,Langlois:2021aji,Hui:2021cpm}).\footnote{If the scalar background is constant or, in particular, vanishing, it is always possible to decouple the equations for the perturbations \cite{Tattersall:2017erk}. Another example where it is possible to diagonalize the potential matrix $\mathbf{V}$ in \eqref{eq:general_system} through a coordinate independent linear transformation is given by charged dilaton black holes \cite{Holzhey:1991bx,Ferrari:2000ep}.} Additional examples of coupled systems of equations arise in the study of the hydrodynamics of strongly-coupled gauge theories through the gauge/gravity duality, see for instance~\cite{Benincasa:2005iv,Son:2006em}.

In the following, we shall consider an explicit example for the potential matrix $\mathbf{V}$ in \eqref{eq:general_system}, and we shall use our analytical approach to compute explicitly the quasinormal modes, and compare them against the values obtained with numerical methods. We stress that the choice of example is non-generic and rather peculiar, and is selected to push the approximation to its extreme condition. It does not correspond to the typical situation, but should be instead interpreted as a worst case scenario. This explicit system is simple enough that it allows analytic control, yet, as we shall discuss, it reduces to well-known results in the $\varepsilon=1$ limit, allowing us to explore the regime of large $\varepsilon$ and to cross-check our computations with previous results in the literature.\footnote{ To test our method on explicit examples of scalar-tensor theories, like that of a scalar-Gauss--Bonnet model, we would need the explicit solution of the background. However, this is known analytically only perturbatively in the coupling; as a consequence, it would allow us to test only the small $\varepsilon$ regime.} In addition, it is an instructive example which allows us to discuss possible issues associated with the method in the case in which the potential is not in a hermitian form.

The toy model we shall focus on is
\begin{equation}
\mathbf{V} (r) = f(r) \begin{pmatrix}
-\frac{8r_s}{r^3} + \frac{\ell^2+\ell+4}{r^2} & \frac{\varepsilon\left(\ell^2+\ell-2\right) (2 r-3r_s)}{r^3}
\\
 \frac{2\varepsilon}{r^2} & \frac{r_s}{r^3} + \frac{\ell^2+\ell-2}{r^2}
\end{pmatrix} \, ,
\qquad
f(r) \equiv 1- \frac{r_s}{r} \, ,
\label{Vtoy}
\end{equation}
where the radial coordinate $r$ is related to the tortoise variable $r_\star$ in \eqref{eq:general_system} through the relation $\D r_\star/\D r=1/f(r)$. It can be checked that this potential satisfies the adiabatic condition of Appendix~\ref{app:adiabatic}, away from level crossing.
$\varepsilon$ is a control parameter that we shall vary within the interval $[0,1]$. This will allow us to explore different cases, from a regime of perturbative mixing ($\varepsilon\ll1$) to a regime where the coupling is order one at $r\sim r_s$ ($\varepsilon\sim 1$). 
Note that, when $\varepsilon=1$, eq.~\eqref{Vtoy} recovers the potential of parity-odd partially massless spin-$2$ fields on a Schwarzschild-de Sitter spacetime in the limit of zero cosmological constant (see, e.g., eq.~(4.13) of Ref.~\cite{Rosen:2020crj}). In this limit, the quasinormal modes are expected to recover the odd spectra of massless spin-$1$ and spin-$2$ fields, as we shall check explicitly for various values of $\ell$. The toy example~\eqref{Vtoy} will also give us the opportunity to discuss possible subtleties and limitations of our approach. 

Our results are summarized in Tables~\ref{tab:qnms1} and \ref{tab:qnms2}, where we compare the quasinormal modes obtained from our analytic method, including both mixing and anharmonic corrections, with the ones computed numerically.\footnote{For the numerical values we used the direct integration method, which consists in solving the system of equations numerically from each of the boundaries and choosing the frequencies to match the left and right solutions (see for example~\cite{Pani:2013pma}).} 

\begin{table}[t]
\vspace{30pt}
\centering
\begin{tabular}{lccccccc}
\firsthline
& \multicolumn{3}{c}{ $\text{Re}(\omega r_s)$ } & & \multicolumn{3}{c}{$-\text{Im}(\omega r_s)$ } \\
\cline{2-4}\cline{6-8}
$\varepsilon$ & Analytic & Numerical & \% & & Analytic & Numerical & \% \\
\hline
0.1 & 0.847 & 0.849 & 0.2 & & 0.183 & 0.176 & 4.0 \\
0.2 & 0.852 & 0.854 & 0.2 & & 0.184 & 0.176 & 4.5 \\
0.3 & 0.861 & 0.861 & $<0.1$ & & 0.185 & 0.175 & 5.7 \\
0.4 & 0.871 & 0.868 & 0.3 & & 0.186 & 0.175 & 6.3 \\
0.5 & 0.882 & 0.876 & 0.7 & & 0.188 & 0.175 & 7.4 \\
0.6 & 0.893 & 0.883 & 1.1 & & 0.190 & 0.176 & 8.0 \\
0.7& 0.905 & 0.891 & 1.6 & & 0.192 & 0.178 & 7.9 \\
0.8& 0.917 & 0.898 & 2.1 & & 0.194 & 0.181 & 7.2 \\
0.9& 0.929 & 0.906 & 2.5 & & 0.195 & 0.185 & 5.4 \\
1.0& 0.941 & 0.915 & 2.8 & & 0.196 & 0.190 & 3.2 \\
\lasthline
\end{tabular}
\caption{\it Comparison between the analytic estimate and the numerical computation for the first set of QNMs for $(n=0,\ell=2)$ perturbations, in the toy model described by the potential in eq.~\eqref{Vtoy}. We include both the mixing corrections of Sec.~\ref{sec:mixing} and the anharmonic (one-loop) corrections of Sec.~\ref{sec:anharmonic}. This set of quasinormal modes corresponds to the local maximum located in the region $r/r_s > 3/2$, where the eigenvalues of $\mathbf{V}(r)$ are real and the approximation scheme we use is justified. }
\label{tab:qnms1}
\end{table}

\begin{table}[t]
\centering
\begin{tabular}{lccccccc}
\firsthline
& \multicolumn{3}{c}{ $\text{Re}(\omega r_s)$ } & & \multicolumn{3}{c}{$-\text{Im}(\omega r_s)$ } \\
\cline{2-4}\cline{6-8}
$\varepsilon$ & Analytic & Numerical & \% & & Analytic & Numerical & \% \\
\hline
0.1 & 0.801 & 0.797 & 0.5 & & 0.213 & 0.194 & 9.8 \\
0.2 & 0.799 & 0.791 & 1.0 & & 0.260 & 0.194 & 35 \\
0.3 & 0.789 & 0.783 & 0.8 & & 0.349 & 0.195 & 79 \\
0.4 & 0.734 & 0.774 & 5.2 & & 0.431 & 0.195 & 121 \\
0.5 & 0.463 & 0.766 & 40 & & 0.428 & 0.194 & 121 \\
0.6 & & 0.758 & & & & 0.194 \\
0.7& & 0.752 & & & & 0.192 \\
0.8& & 0.748 & & & & 0.189 \\
0.9& & 0.746 & & & & 0.184 \\
1.0& & 0.747 & & & & 0.178 \\
\lasthline
\end{tabular}
\caption{\it Comparison between the analytic estimate and the numerical computation for the second set of QNMs for $(n=0,\ell=2)$ perturbations, in the toy model described by the potential in eq.~\eqref{Vtoy}. We include both the mixing corrections of Sec.~\ref{sec:mixing} and the anharmonic (one-loop) corrections of Sec.~\ref{sec:anharmonic}. This set of quasinormal modes corresponds to the local maximum located in the region $1< r/r_s < 3/2$, where the eigenvalues of $\mathbf{V}(r)$ are complex close to level crossing or for large $\varepsilon$. The approximation scheme is not justified and the results for the imaginary part of the quasinormal frequencies are off for sizeable mixing parameter $\varepsilon$, especially for the imaginary part.
}
\label{tab:qnms2}
\end{table}

\begin{table}[t]
\vspace{50pt}
\centering
\begin{tabular}{lccccccc}
\firsthline
& \multicolumn{3}{c}{ $\text{Re}(\omega r_s)$ } & & \multicolumn{3}{c}{$-\text{Im}(\omega r_s)$ } \\
\cline{2-4}\cline{6-8}
$\ell$ & Analytic & Numerical & \% & & Analytic & Numerical & \% \\
\hline
2 & 0.941 & 0.915 & 2.8 & & 0.196 & 0.190 & 3.2 \\ 
3 & 1.353 & 1.314 & 3.0 & & 0.203 & 0.191 & 6.2 \\ 
4 & 1.751 & 1.706 & 2.6 & & 0.206 & 0.192 & 7.4 \\ 
5 & 2.143 & 2.096 & 2.3 & & 0.209 & 0.192 & 8.9 \\ 
6 & 2.532 & 2.484 & 1.9 & & 0.210 & 0.192 & 9.3 \\ 
7 & 2.920 & 2.871 & 1.7 & & 0.211 & 0.192 & 9.8 \\ 
8 & 3.307 & 3.258 & 1.5 & & 0.212 & 0.192 & 10.3 \\ 
9 & 3.693 & 3.644 & 1.3 & & 0.212 & 0.192 & 10.3 \\ 
10 & 4.079 & 4.030 & 1.2 & & 0.213 & 0.192 & 10.8 \\ 
11 & 4.465 & 4.416 & 1.1 & & 0.213 & 0.192 & 10.7 \\ 
12 & 4.850 & 4.802 & 1.0 & & 0.213 & 0.192 & 10.7 \\ 
13 & 5.235 & 5.188 & 0.9 & & 0.214 & 0.192 & 11.2 \\ 
14 & 5.620 & 5.573 & 0.8 & & 0.214 & 0.192 & 11.2 \\ 
15 & 6.005 & 5.958 & 0.8 & & 0.214 & 0.192 & 11.2 \\ 
16 & 6.390 & 6.344 & 0.7 & & 0.214 & 0.192 & 11.2 \\ 
17 & 6.775 & 6.729 & 0.7 & & 0.214 & 0.192 & 11.2 \\ 
18 & 7.160 & 7.114 & 0.6 & & 0.214 & 0.192 & 11.2 \\ 
19 & 7.545 & 7.500 & 0.6 & & 0.214 & 0.192 & 11.2 \\ 
20 & 7.930 & 7.885 & 0.6 & & 0.214 & 0.192 & 11.2 \\ 
21 & 8.315 & 8.270 & 0.5 & & 0.214 & 0.192 & 11.2 \\ 
22 & 8.699 & 8.655 & 0.5 & & 0.214 & 0.192 & 11.2 \\ 
23 & 9.084 & 9.040 & 0.5 & & 0.214 & 0.192 & 11.2 \\ 
24 & 9.469 & 9.425 & 0.5 & & 0.214 & 0.192 & 11.2 \\ 
25 & 9.854 & 9.810 & 0.4 & & 0.214 & 0.192 & 11.2 \\ 
26 & 10.238 & 10.195 & 0.4 & & 0.214 & 0.192 & 11.2 \\ 
27 & 10.623 & 10.580 & 0.4 & & 0.214 & 0.192 & 11.2 \\ 
28 & 11.008 & 10.966 & 0.4 & & 0.214 & 0.192 & 11.2 \\ 
29 & 11.393 & 11.351 & 0.4 & & 0.214 & 0.192 & 11.2 \\ 
30 & 11.777 & 11.736 & 0.4 & & 0.214 & 0.192 & 11.2 \\ 
\end{tabular}
\caption{\it Comparison between the analytic estimate and the numerical computation for the fundamental QNM $n=0$ of a spin-1 field on a Schwarzschild spacetime computed from the potential in eq.~\eqref{Vtoy} with $\varepsilon=1$, as a function of $\ell$. We include both the mixing corrections of Sec.~\ref{sec:mixing} and the anharmonic (one-loop) corrections of Sec.~\ref{sec:anharmonic}. The numerical values are from Ref.~\cite{2006MPLA...21.2671P}. }
\label{tab:qnmseikonal}
\end{table}

Some comments are in order here. Let us start by looking at Table~\ref{tab:qnms1}, where we report the first set of frequencies with $n=0$ and $\ell=2$. Note that, as anticipated above, when $\varepsilon=1$ we recover the quasinormal modes of spin-1 fields on a Schwarzschild spacetime \cite{Brito:2013yxa,Rosen:2020crj,2006MPLA...21.2671P}. The precision of the analytic estimates is quite remarkable, with an error less than a few percent in most cases and $\lesssim \mathcal{O}(0.1\%)$ for the real part of $\omega$ for small values of $\varepsilon$. The percent error with respect to the numerical values can be further reduced by including higher orders corrections. The case of the second set of modes, reported in Table~\ref{tab:qnms2}, is instead slightly different. As it can be seen from the table, the error with respect to the numerical values becomes large as $\varepsilon$ increases, signaling a breakdown of the analytic approximation. The reason for this is that, in the particular example under consideration, the eigenvalues of the potential \eqref{Vtoy} acquire a nonzero imaginary part on a finite range of $r$. For small values of $\varepsilon$, this range is small and the local maxima are located in the region where the potential eigenvalues are real. The approximation is therefore well defined and provides an accurate estimate of both the two sets of quasinormal frequencies for small $\varepsilon$. However, as $\varepsilon$ increases, the maximum of one of the two eigenvalues of $\mathbf{V}(r)$ falls close to the boundary of the region where the eigenvalue is complex, approaching the edge until it is lost. The approximation scheme applied to this eigenvalue is no longer justified for sizeable values of $\varepsilon$ and cannot be used to estimate the corresponding set of quasinormal normal modes.\footnote{The maximum of the other eigenvalue falls instead always within the interval where the eigenvalue is real, for all values of $\varepsilon\in [0,1]$. This is why all the quasinormal modes in Table~\ref{tab:qnms1} can be reliably estimated, and with very good accuracy, using our approximation scheme. } This explains why the frequencies in Table~\ref{tab:qnms2} are significantly off for sizeable mixing parameter $\varepsilon$, especially for the imaginary part.

We should stress again that what happens to the modes in Table~\ref{tab:qnms2} is \textit{not} the typical situation. In fact, as we show explicitly in Appendix~\ref{app:Schreqs}, when an action formulation is available it is in general possible to recast the $N$ coupled equations for the perturbations in a Schr\"odinger-like form~\eqref{eq:general_system} with hermitian potential. As a result, the eigenvalues of the new $\mathbf{V} (r)$ are real and the breakdown shown in Table~\ref{tab:qnms2} does not happen.
In particular, if $N=2$ and the potential is frequency independent, the required transformation can be always found analytically and the hermitian potential can be written in closed form.\footnote{This is in agreement with the results of Ref.~\cite{Rosen:2020crj} (see in particular Appendix D), where this is shown explicitly for the particular case of the system \eqref{Vtoy} with $\varepsilon=1$.}
However, we find the specific example \eqref{Vtoy} still particularly instructive. Even if not generic, it is useful to show what the limitations of our analytical approximation scheme might be when a closed-form expression for the hermitian potential is not available, or when an action formulation is not available and our proof in Appendix~\ref{app:Schreqs} does not apply. For a generic two by two potential matrix
\begin{equation}
\mathbf{V} = \begin{pmatrix}
a & \hspace{5pt} b \\
c & d \\
\end{pmatrix},
\end{equation}
it is easy to show that the eigenvalues are always real when the off-diagonal elements have the same sign, and can develop an imaginary part only when $-  b\cdot c  > (a-d)^2 /8$. This explains the properties previously emphasized for the QNMs of our toy example~\eqref{Vtoy}, since the quantity $b\cdot c$ changes sign for $r/r_s < 3/2$ and its size is controlled by $\varepsilon$, elucidating the results of Tables~\ref{tab:qnms1} and~\ref{tab:qnms2}.

In Table~\ref{tab:qnmseikonal} we compare, as a function of $\ell$ and for $\varepsilon=1$, the analytic computation for the $n=0$ quasinormal mode associated with the local maximum located in the region $r/r_s > 3/2$, with the numerical value of the $n=0$ quasinormal frequency of a spin-1 field on a Schwarzschild spacetime \cite{2006MPLA...21.2671P}. As discussed previously, in this limit the system~\eqref{Vtoy} recovers the odd QNM spectrum of massless spin-1 and spin-2 fields on a Schwarzschild background and the spin-1 field is the one associated with the set of modes for which our approximation scheme is valid. The accuracy for the real part is remarkable and reaches the sub-percent level in the eikonal limit, whereas for the imaginary part it asymptotes to around $10\%$, signaling a systematic shift probably induced by the approximate treatment of the off-diagonal contributions, which are sizeable for $\varepsilon=1$.
Note that this can well be a feature of the field basis chosen to write the potential. In some cases convergence can be improved by writing the potential in a hermitian form, using the general procedure outlined in Appendix~\ref{app:Schreqs} below. An explicit example of this is precisely given by partially massless spin-2 fields on Schwarzschild-de Sitter spacetime: in \cite{Rosen:2020crj}, it is shown that, in the basis where the potential is hermitian, the off-diagonal mixing terms go as $\sim 1/\ell$ in the large-$\ell$ limit. We thus expect in this case a better precision of the estimate of the imaginary part of the frequencies in the eikonal limit. 

We conclude mentioning that a different approach to estimate the QNMs in the eikonal limit for a system of coupled linear equations has been discussed in Refs.~\cite{Glampedakis:2019dqh,Silva:2019scu,Bryant:2021xdh}. We stress though that our method is more general as it applies, under the assumptions discussed above, to any value of $\ell$.


\section{Discussion}
\label{sec:discussion}

In this work we have introduced a new analytic approximation scheme to estimate the QNM frequencies for systems of coupled linear differential equations of second order in Schr\"odinger-like form.
Examples where coupled equations arise include massive and partially massless spinning fields on Schwarzschild-(anti-)de Sitter spacetimes, and perturbations around black holes beyond GR, e.g.~in the context of scalar-tensor theories. Moreover, as proved in Appendix~\ref{app:Schreqs}, whenever an action formulation is available the system of equations can be put in Schr\"odinger-like form with a hermitian potential.

The approximation scheme is always under control when the mixing corrections are perturbative and small. However, we stress that our method can be applied equally well to cases where the mixing is large and provides up to order-one corrections close to the black hole, as long as an emergent perturbativity condition for the expansion of the potential around a critical point is satisfied, as discussed at the end of Section~\ref{sec:mixing}. Another main advantage of this approach with respect, for instance, to more phenomenological parameterizations of the ringdown in theories beyond GR (see e.g.~Ref.~\cite{McManus:2019ulj}) is that it allows to extract the quasinormal frequencies from the shape of the potential matrix around the maxima of its eigenvalues. As a consequence, with a finite experimental accuracy, the knowledge of the full shape of the potential is not needed to make accurate predictions.\footnote{This is in line with the results  of Ref.~\cite{Volkel:2022khh}, where it is shown that, under the assumption that corrections to the general relativistic potential are captured by a (finite) power series in $r_s/r$, only the form of the potential near the light ring can be robustly constrained via reconstruction  techniques from a modified spectrum of QNMs.} In addition, the potential does not need to be expressed as a series of (inverse) power corrections to the GR potential, but can have an arbitrary functional form.
The only assumptions are that $\mathbf{V}(r_\star)$ in \eqref{eq:general_system} is a smooth function of the radial tortoise coordinate $r_\star$, satisfying the fall-off condition~\eqref{eq:fall_off} and having real eigenvalues.
In fact, our method is even more general as it 
applies even to cases where the potential depends also generically on the frequency, $\mathbf{V}\equiv \mathbf{V}(r_\star,\omega)$, and has real eigenvalues for real values of $\omega$ --- see Refs.~\cite{Rosen:2020crj,Hui:2021cpm} for some explicit examples. In such cases, one can formally go through the same procedure described in Section~\ref{sec:approximation}, regarding the various quantities --- such as for instance the potential eigenvalues --- as functions of the frequency, and at the end solve self-consistently for the quasinormal modes. 

The main results of our analysis are the leading order quantization condition~\eqref{eq:leading_qnm}, the quadratic mixing correction~\eqref{eq:mixing_qnm} and the one-loop anharmonic correction~\eqref{eq:oneloop_qnm}. The analysis of anharmonic corrections could be easily extended at higher loop orders using the methods of~\cite{Bender:1969si,Sulejmanpasic:2016fwr}. It would be interesting to extend these methods to the analysis of a system of equations and investigate the Borel summability of the perturbative series along the lines of~\cite{Hatsuda:2019eoj}.

Our method is well justified when the eigenvalues of the potential matrix are real, and when the imaginary part of the frequency $\omega$ is small compared to its real part. The first condition can be always accommodated (for real values of the frequency), as long as the potential matrix is chosen to be hermitian, which can be always done when an action formulation is available, as previously discussed. In order to validate our approach and show its limitations we considered an example modeled on the case of partially massless spin-$2$ fields on a Schwarzschild-de Sitter spacetime in the limit of zero cosmological constant studied in Ref.~\cite{Rosen:2020crj}, using a non-hermitian potential. The accuracy, when the approximation is justified, is remarkable and it is expected to improve when using the hermitian formulation. 

The fact that our approximation scheme does not depend on the full shape of the potential away from the maxima of the eigenvalues makes it particularly suitable when combined with an effective field theory approach \cite{Franciolini:2018uyq,Hui:2021cpm}. It would be interesting to generalize the light-ring expansion introduced in \cite{Franciolini:2018uyq} for a single equation to coupled systems, and understand more systematically how to connect possible deviations in the observed ringdown frequencies to the effective couplings using our approximation scheme. We leave this and related aspects for future work.

\acknowledgments We thank Riccardo Penco for discussions at the early stages of this work. AP acknowledges support from the Simons Foundation Award No. 658906 and the grant DOE DE-SC0011941. ET is partly supported by the Italian MIUR under contract 2017FMJFMW (PRIN2017). LH acknowledges support by the DOE DE-SC0011941 and a Simons Fellowship in Theoretical Physics.

\appendix

\section{Perturbation equations in Schr\"odinger-like form}
\label{app:Schreqs}
Let us consider the following quadratic action for a $N$-dimensional field vector $\vec{\phi}$, describing the perturbations around a given stationary background:
\begin{equation}
S= \sum_{\ell m}\int \frac{\D \omega}{2\pi}\int \D r_\star \left[
-A_{ij} \partial_{r_\star}\hat \phi_i \partial_{r_\star}\hat \phi_j^* + \frac{1}{2} \left( B_{ij}  \partial_{r_\star}  \hat \phi_i  \, \hat \phi_j^* - B_{ij}\hat\phi_i  \partial_{r_\star}\hat \phi_j^*  \right) +C_{ij}\hat\phi_i \hat \phi_j^*
\right] \, ,
\label{generalEFT}
\end{equation}
where we Fourier transformed in time and decomposed the field in spherical harmonics as $\vec{\phi}(t,r_\star,\theta,\varphi)= \sum_{\ell m}\int \frac{\D \omega}{2\pi} e^{-\I\omega t}Y_{\ell m}(\theta,\varphi)\vec{\hat{\phi}}_{\ell,m}(\omega,r_\star)$,\footnote{In \eqref{generalEFT} and below, we shall drop for ease of notation the ${\ell,m}$ indices on the Fourier transformed fields $\hat \phi_i$. For simplicity, we also assume spherical symmetry for the background: this allows to decouple the quadratic Lagrangians of modes corresponding to different $\ell$'s. The analysis is valid, with minor modifications, also in the case of slowly rotating backgrounds, as in~\cite{Hui:2021cpm}.}
and where $\mathbf{A}$ and $\mathbf{C}$ are generic hermitian matrices, while $\mathbf{B}$ is anti-hermitian. Note that the quadratic actions of massive and partially massless spinning fields on Schwarzschild spacetimes (see, e.g., Appendix D of Ref.~\cite{Rosen:2020crj}), as well as the effective action of perturbations around black hole solutions in scalar-tensor theories (see, e.g., Refs.~\cite{Franciolini:2018uyq,Hui:2021cpm}) can all in general be put in the form \eqref{generalEFT}, after making an appropriate choice of gauge and solving for the constraints.
In \eqref{generalEFT}, $\mathbf{A}$, $\mathbf{B}$ and $\mathbf{C}$ are generic functions of the radial (tortoise) coordinate $r_\star$, while $\mathbf{B}$ and $\mathbf{C}$ may also depend on the frequency $\omega$ (as well as on the spherical harmonics quantum numbers).

First of all we note that the absence of gradient instabilities requires $A_{ij}$ to be also positive definite (in addition to being hermitian). It is thus possible to diagonalize the kinetic term $A_{ij} \partial_{r_\star} \hat \phi_i \partial_{r_\star}\hat \phi_j $ in the action \eqref{generalEFT}, via an $r_\star$-dependent similarity transformation $\vec{\hat\phi}\rightarrow \mathbf{N}(r_{\star}) \cdot \vec{\hat\phi}$. Moreover, by canonically normalizing the new field variables we can bring the matrix to the identity: $A_{ij}\rightarrow \delta_{ij}$.
As a result, the action can be recast as in \eqref{generalEFT} with $A_{ij}\equiv \delta_{ij}$, and where, up to integrations by parts, $\mathbf{B}$ and $\mathbf{C}$ can be chosen to be respectively anti-hermitian and hermitian.

The equations of motion take the following general form:
\begin{equation}
\left( \delta_{ij} \partial_{r_\star}^2 + \mathbf{B}_{ij}\partial_{r_\star} + \tilde{\mathbf{C}}_{ij}\right) \hat\phi_j =0 \, ,
\label{eomgen}
\end{equation}
where $\tilde{\mathbf{C}} = \frac{1}{2}\partial_{r_\star}\mathbf{B}+\mathbf{C}$. From the symmetry properties of $\mathbf{B}$ and $\mathbf{C}$, it follows that $\tilde{\mathbf{C}}$ will be in general neither hermitian nor anti-hermitian.
We now introduce the matrix
\begin{equation}
\mathbf{M}(u) = \mathcal{P} \left\{e^{-\frac{1}{2}\int^{u}_{\bar u} \mathbf{B}(r)\D r} \right\}\, ,
\label{eq:solMmatrix}
\end{equation}
where we used the path-ordered exponential $\mathcal{P}\{\exp(\cdot)\}$ defined by
\begin{equation}
\mathcal{P} \left\{e^{-\frac{1}{2}\int^{u}_{\bar u} \mathbf{B}(r)\D r} \right\} 
\equiv \sum_{n=0}^\infty \frac{(-1)^n}{2^n}\int^{u}_{\bar u} \D r_n \int^{r_n}_{\bar u} \D r_{n-1}\cdots \int^{r_2}_{\bar u} \D r_1 \, 
\mathbf{B}(r_n)\mathbf{B}(r_{n-1})\cdots \mathbf{B}(r_1) \, .
\end{equation}
Thanks to the anti-hermiticity of $\mathbf{B}$, the matrix $\mathbf{M}(u)$ can be shown to be unitary. To see this, consider first the conjugate matrix $\mathbf{M}^\dagger(u)$. From the definition in terms of power series, the properties of the transpose operator and the anti-hermiticity of $\mathbf{B}$, it follows that 
\begin{equation}
\mathbf{M}^\dagger(u) = \bar{\mathcal{P}} \left\{e^{\frac{1}{2}\int^{u}_{\bar u} \mathbf{B}(r)\D r} \right\}\, ,
\label{eq:Mtranspose}
\end{equation}
where $\bar{\mathcal{P}}$ denotes \emph{anti}-path-ordering.
Consider now the matrix $\mathbf{Z}(u)=\mathbf{M}(u)\mathbf{M}^\dagger(u)$. By differentiating with respect to $u$ we find that it satisfies the differential equation $\mathbf{Z}'(u) = -\frac{1}{2} \left(\mathbf{B}(u)\mathbf{Z}(u)-\mathbf{Z}(u)\mathbf{B}(u)\right)$, and by construction $Z_{ij}(\bar{u})=\delta_{ij}$. It is easy to check that $Z_{ij}(u)=\delta_{ij}$ for every $u$ solves the differential equation and satisfies the initial condition, so that by the Cauchy-Lipschitz-Picard uniqueness theorem it follows that $\mathbf{Z}(u)=\mathbf{1}$ and $\mathbf{M}^\dagger(u)= \mathbf{M}^{-1}(u)$, \emph{i.e.} $\mathbf{M}$ is unitary.

Performing the field redefinition
\begin{equation}
\vec{\hat\phi} = \mathbf{M}(r_\star) \cdot \vec{\psi} \, ,
\end{equation}
and multiplying on the left by $\mathbf{M}^{-1}$, the equation of motion \eqref{eomgen} can be recast in the Schr\"odinger-like form~\eqref{eq:general_system}, where 
\begin{equation}
\mathbf{V}=\omega^2 - \mathbf{M}^{-1} \left(\mathbf{C}-\frac{1}{4} \mathbf{B}^2 \right) \mathbf{M} \,,
\end{equation}
is in general a function of $r_\star$, $\omega$ and $\ell$. Moreover, thanks to the unitarity of $ \mathbf{M}$, the anti-hermiticity of $ \mathbf{B}$ and the hermiticity of $ \mathbf{C}$, the potential $\mathbf{V}$ is manifestly hermitian.

In general it might be difficult to find the explicit expression for the matrix $\mathbf{M}$. The matter simplifies significantly, however, in the special case in which the matrix $\mathbf{B}$ is anti-symmetric, and $N=2$ (a coupled system of two equations). In this case the matrix $\mathbf{M}$ is orthogonal, and for $N=2$ it can be parametrized in terms of a single angle $\beta(r_\star)$. Performing the most general orthogonal transformation and requiring the cancellation of the term proportional to $\partial_{r_\star}\psi$ gives a first order differential equation for $\beta(r_\star)$ that can be always solved explicitly to obtain a closed form expression for the hermitian potential $\mathbf{V}$.

\section{Proof of the generalized adiabatic theorem}
\label{app:adiabatic}

In this appendix we present the detailed proof of the generalized adiabatic theorem described in section~\ref{sec:adiabatic}.
The proof follows the standard approach of the quantum mechanical adiabatic theorem (see for instance~\cite{Sakurai:2011zz}), but does not rely on the hermiticity of the evolution operator.

\begin{theorem}
Consider a finite dimensional system described by a time-dependent Schr\"odinger equation with (possibly non-Hermitian) evolution operator $\mathbf{K}(\tau)$. Assume moreover that the spectrum of $\mathbf{K}$ is real and non-degenerate for a range of $\tau\in [\tau_{i},\tau_{f}]$, with instantaneous eigenvalues $\lambda_{a}(\tau)$ in increasing order for $a=1,\dots,2N$ and corresponding eigenstates $\lvert k_{a}(\tau) \rangle$. Then, in the limit $\partial_{\tau}\mathbf{K}(\tau)\rightarrow 0$ and away from level crossing, if the system at time $\tau_{i}$ is in the instantaneous eigenstate $\lvert k_{a}(\tau_{i}) \rangle$ of $\mathbf{K}(\tau_{i})$ for some fixed $a$, at time $\tau_{f}$ it will be in the eigenstate $\lvert k_{a}(\tau_{f}) \rangle$, up to a scalar factor.
\end{theorem}
\begin{proof}
Consider the basis of instantaneous eigenstates of $\mathbf{K}(\tau)$:
\begin{equation}
\label{eq:inst_eig}
\mathbf{K}(\tau) \lvert k_{a}(\tau) \rangle = \lambda_{a}(\tau) \lvert k_{a}(\tau) \rangle.
\end{equation}
We decompose the state $\lvert \psi(\tau) \rangle$ obtained from the time evolution with prescribed initial condition in terms of instantaneous eigenstates, as follow:
\begin{equation}
\label{eq:decomposition}
\lvert \psi(\tau) \rangle = \sum_{a} c_{a}(\tau) e^{-\mathrm{i} \phi_{a}(\tau)} \lvert k_{a}(\tau) \rangle,
\end{equation}
where the phase factor $\phi_{a}$ has been introduced for convenience and is defined by 
\begin{equation}
\phi_{a}(\tau)= \int_{\tau_{i}}^{\tau} \lambda_{a}(\tau') {\rm d}\tau'.
\end{equation}
The Schr\"odinger equation $\mathrm{i}\partial_{\tau}\lvert \psi(\tau) \rangle = \mathbf{K}(\tau) \lvert \psi(\tau) \rangle$ takes now the simple form
\begin{equation}
\sum_{a} e^{-\mathrm{i} \phi_{a}(\tau)} \left( \dot{c}_{a}(\tau) + c_{a} \partial_{\tau} \right) \lvert k_{a}(\tau) \rangle=0.
\end{equation}
Projecting on the basis of instantaneous eigenstates by acting with the equal time dual vectors $\langle \tilde{k}_{b}(\tau) \rvert$ it follows that
\begin{equation}
\dot{c}_{b}(\tau) = - \sum_{a} c_{a}(\tau) e^{-\mathrm{i}(\phi_{a}-\phi_{b})} \langle \tilde{k}_{b}(\tau) \rvert \partial_{\tau}\lvert k_{a}(\tau) \rangle.
\end{equation}
It is now useful to separate the cases $a=b$ and $a\neq b$.
For $a\neq b$, by taking the time derivative of eq.~\eqref{eq:inst_eig}, projecting by acting with the dual vectors $\langle \tilde{k}_{b}(\tau) \rvert$ and using eq.~\eqref{eq:left_eig} we obtain
\begin{equation}
\langle \tilde{k}_{b}(\tau) \rvert \dot{\mathbf{K}}(\tau) \lvert k_{a}(\tau) \rangle = \left( \lambda_{a}- \lambda_{b} \right) \langle \tilde{k}_{b}(\tau) \rvert \partial_{\tau}\lvert k_{a}(\tau) \rangle,
\end{equation}
where $\dot{\mathbf{K}}(\tau)= \partial_{\tau}\mathbf{K}$.
We arrive at:
\begin{equation}
\dot{c}_{b}(\tau) = - \langle \tilde{k}_{b}(\tau) \rvert \partial_{\tau}\lvert k_{b}(\tau) \rangle \; c_{b}(\tau) - \sum_{a\neq b} c_{a}(\tau) e^{-\mathrm{i}(\phi_{a}-\phi_{b})} \dfrac{\langle \tilde{k}_{b}(\tau) \rvert \dot{\mathbf{K}}(\tau) \lvert k_{a}(\tau) \rangle}{\lambda_{a}-\lambda_{b}}.
\end{equation}
Since the spectrum is non-degenerate and the phase factors $\phi_{a}$ are real (thanks to the reality of the eigenvalues $\lambda_{a}$), the terms $a\neq b$ can be safely neglected in the limit $\partial_{\tau}\mathbf{K}(\tau)\rightarrow 0$, so that we obtain:
\begin{equation}
\dot{c}_{b}(\tau) = - \langle \tilde{k}_{b}(\tau) \rvert \partial_{\tau}\lvert k_{b}(\tau) \rangle \; c_{b}(\tau),
\end{equation}
which implies $c_{b}(\tau) = e^{\mathrm{i} \theta_{b}(\tau)} c_{b}(\tau_{i})$.
By plugging $c_{a}(\tau)$ in eq.~\eqref{eq:decomposition} we arrive at the desired result.
\end{proof}
The term $e^{\mathrm{i} \theta_{b}(\tau)}$ is an adiabatic factor that in the quantum mechanical case, corresponding to Hermitian $\mathbf{K}(\tau)$, can be shown to be a pure phase and gives rise to Berry's phase when a parameter of the system is varied adiabatically and cyclically.

\section{Asymptotic approximation for the mixing correction}
\label{app:asymptotic}

The function 
\begin{equation}
F_{n}(y)= \left(\dfrac{1}{2^{n-1}\cdot n!}\right) \dfrac{1}{y} \int_{0}^{+\infty} P_{n}(\sigma) e^{-\sigma^{2}} e^{-2 y \sigma} {\rm d}\sigma
\end{equation}
appearing in the mixing correction to the quasinormal modes in equation~\eqref{eq:mixing_qnm} admits an asymptotic expansion in the limit $|2y| \rightarrow \infty$ under the assumption that $|{\rm arg}(2y)|<\frac{\pi}{2}$. Indeed, expanding the function $P_{n}(\sigma)e^{-\sigma^{2}}$ in powers of $\sigma$ around $\sigma=0$, $P_{n}(\sigma)e^{-\sigma^{2}}= \sum_{k=0}^{\infty} a_{n,k} \sigma^{k}$, and applying Watson's lemma (see \emph{e.g.} Theorem 15.2.8 of~\cite{Simon:2015}) it follows that in the limit $|2y| \rightarrow \infty$ (for $|{\rm arg}(2y)|<\frac{\pi}{2}$):
\begin{equation}
F_{n}(y) \sim \left(\dfrac{1}{2^{n-1}\cdot n!}\right) \dfrac{1}{y} \sum_{k=0}^{\infty} a_{n,k} \dfrac{k!}{(2y)^{k+1}}.
\end{equation}
The polynomial $P_{n}(\sigma)$ has only even powers of $\sigma$ and has degree $2n+2$:
\begin{equation}
P_{n}(\sigma)= P_{n,0} + P_{n,2} \sigma^{2} + \dots + P_{n,2n+2} \sigma^{2n+2},
\end{equation} 
therefore the coefficients $a_{n,k}$ with odd $k$ are all identically zero and the first few coefficients $a_{n,k}$ can be expressed in term of the $P_{n,k}$ as follows:
\begin{equation}
a_{n,0}= P_{n,0}, \qquad a_{n,2}= P_{n,2}-P_{n,0}, \qquad a_{n,4}= P_{n,4}-P_{n,2}+ P_{n,0}/2, \qquad \dots
\end{equation}
The leading term of the asymptotic expansion is proportional to $P_{n,0}$, which can be computed explicitly from the definition of $P_{n}(\sigma)$ to be
\begin{equation}
a_{n,0}=P_{n,0}= \dfrac{1}{\sqrt{\pi}} \int_{-\infty}^{+\infty} \rho^{2} H_{n}^{2}(\rho) e^{-\rho^{2}} {\rm d}\rho = 2^{n} (n!) \left(n+\dfrac{1}{2}\right),
\end{equation}
from which it follows that
\begin{equation}
F_{n}(y) \sim \left(n+\dfrac{1}{2}\right) \dfrac{1}{y^{2}} + \dots ,
\end{equation}
so that in this limit the mixing correction to the quasinormal modes is approximately:
\begin{equation}
\left(\delta\omega^{2}_{n}\right)^{(2)} = \varepsilon^2 \dfrac{b_{1}c_{1}}{a_{2}} F_{n}\left(\dfrac{\alpha_{n}}{\sqrt[4]{a_{2}}}\right) \simeq \varepsilon^2 \dfrac{b_{1}c_{1}}{a_{2}} \left(n+\dfrac{1}{2}\right) \dfrac{\sqrt{a_{2}}}{\alpha_{n}^{2}}.
\end{equation}
From the definition of $\alpha_{n}^{2}$ we have
\begin{equation}
\dfrac{\alpha_{n}^{2}}{\sqrt{a_{2}}} = - (2n+1) + \I \dfrac{(d_{0}-a_{0})}{\sqrt{a_{2}}}.
\end{equation}
Whether this asymptotic expansion provides an accurate approximation or not depends on the value of the parameters characterizing the potential and should be checked on a case by case basis.

\bibliographystyle{JHEP}
\addcontentsline{toc}{section}{References}

\providecommand{\href}[2]{#2}\begingroup\raggedright\endgroup

\end{document}